\documentclass[11pt,a4paper]{article}
\usepackage[margin=2.3cm]{geometry}
\usepackage[T1]{fontenc}
\usepackage{amsmath,amssymb,amsthm}
\usepackage{booktabs,graphicx}
\usepackage[colorlinks=true,linkcolor=blue,citecolor=blue,urlcolor=blue]{hyperref}

\theoremstyle{plain}
\newtheorem{theorem}{Theorem}
\newtheorem{lemma}[theorem]{Lemma}
\newtheorem{proposition}[theorem]{Proposition}
\newtheorem{corollary}[theorem]{Corollary}
\newtheorem{conjecture}[theorem]{Conjecture}
\newtheorem{convention}[theorem]{Convention}
\newtheorem*{theoremA}{Theorem A}
\newtheorem*{theoremB}{Theorem B}
\newtheorem{problem}{Problem}
\theoremstyle{definition}
\newtheorem{definition}[theorem]{Definition}
\newtheorem{remark}[theorem]{Remark}

\newcommand{\Z}{\mathbb{Z}}
\newcommand{\Cay}{\mathrm{Cay}}

\newcommand{\pp}{\mathrm{pp}}
\newcommand{\indpp}{\mathrm{indpp}}
\newcommand{\ipp}{\mathrm{ipp}}

\title{\textbf{The Cayley Completion of a Graph}}
\author{Rigobert Fokam Souop\thanks{Laboratory of Energy, Signal, Imaging and
Automation (LESIA), University of Ngaound\'er\'e, P.O.\ Box 454, Ngaound\'er\'e,
Cameroon. \texttt{fokamrigobert@gmail.com}}
\and Laurent Bitjoka\thanks{LESIA, University of Ngaound\'er\'e; and Laboratory
of Scientific Artificial Intelligence and Applied Mathematics, University of
Garoua, P.O.\ Box 346, Garoua, Cameroon.
\texttt{laurent.bitjoka@univ-garoua.cm}}}
\date{August 2026}

\begin{document}

\maketitle

\begin{abstract}
A finite connected graph is rarely a Cayley graph. We measure how far it is
from being one: given $G$ with $n$ vertices and $m$ edges, how few edges must
be added, or added and deleted, before the result is a Cayley graph of an
abelian group of order $n$ on the same vertex set? This defines two
invariants, the completion number $\gamma^{+}$, which permits additions only,
and the Cayley edit distance $\gamma_{\triangle}$, which permits both, each
normalized by $m$.

We show that deciding the edit version is NP-complete already for a fixed
cyclic host, by a reduction from Hamiltonian Cycle in which the edit cost of a
labeling is $n+m-2k$ when the labeling realizes a longest path with $k$ edges;
there the optimal cost is $m-n+2\pp(G)$, which the matching number bounds in
polynomial time. We prove
that irregularity alone forces $\gamma^{+}(G)\ge n\Delta^{*}/(2m)-1$, where
$\Delta^{*}$ is the least $d\ge\Delta$ with $nd$ even, a bound computable in
linear time from the degree sequence; we characterize its equality case
exactly. It is attained on the star, where $\gamma^{+}(K_{1,q})=(q-1)/2$, and
the star maximizes $\gamma^{+}$, while $\gamma_{\triangle}$ is bounded by an
absolute constant. We determine paths and grids exactly, obtaining
$\gamma^{+}(P_n)=\gamma^{+}(P_n\,\square\,P_n)=1/(n-1)$, and we show that
$\gamma_{\triangle}(K_{1,q})$ tends to $2$, and not to the value $3/2$
suggested by the additive case.

We then report an exhaustive certified census of all $995$ connected graphs on
at most seven vertices. The degree bound is attained on $89.4\%$ of them and
the two invariants separate strictly on $84.7\%$, though both rates vary
sharply with order: attainment runs $100\%,100\%,84.8\%,89.7\%$ and separation
$0\%,61.9\%,73.2\%,87.7\%$ for $n=4,5,6,7$, so the pooled figures are
dominated by the $853$ graphs on seven vertices. The star is the unique
maximizer of both. We also show that edit count and the bi-Lipschitz
distortion of the completed host are independent invariants, moving in
opposite directions on stars and on paths. Data and certificates are deposited
at \texttt{doi:10.5281/zenodo.21852006}.
\end{abstract}

\medskip
\noindent\textbf{Keywords:} Cayley graph, abelian group, graph editing,
completion number, graph census.\\
\noindent\textbf{MSC 2020:} 05C25, 05C12, 05C85, 68Q17, 20K01.

\section{Introduction}

A finite connected graph is rarely a Cayley graph. One response, taken by
\cite{fokam-p1,fokam-p2}, is to enlarge the ambient object: seek an abelian
group $\Gamma$ and an isometric embedding of $G$ into $\Cay(\Gamma,S)$, paying
for exactness with a host that may be much larger than $G$. This paper takes
the opposite route. We keep the vertex set fixed and perturb the graph,
asking how few edges must be added, or added and deleted, before $G$ becomes
an abelian Cayley graph on its own $n$ vertices. The smallest instance is
familiar: a single added edge turns the path $P_n$ into the cycle $C_n$, and
$C_n$ is $\Cay(\Z_n,\{\pm1\})$. We call the general answer the \emph{Cayley
completion} of $G$ and study two invariants: the \emph{completion number}
$\gamma^{+}$, which permits additions only, and the \emph{Cayley edit
distance} $\gamma_{\triangle}$, which permits both, each normalized by
$|E(G)|$.

Two features make the question non-trivial. Being an abelian Cayley graph on
a fixed vertex set is not a property of the abstract graph alone but of the
graph together with a labeling of its vertices by group elements, so the
problem carries a search over $\Gamma$ and over labelings. And regularity is
necessary but far from sufficient: vertex-transitivity is already strictly
weaker than being a Cayley graph \cite{mckay-praeger1994}, so the invariants
do not collapse on regular graphs.

\paragraph{Main results.}
Our two principal results are the following; they are restated with proofs as
Theorem~\ref{thm:nphard} and Theorem~\ref{thm:degree}.

\begin{theoremA}[Hardness]
Deciding, for a fixed cyclic host $C_n$ and a given budget $B$, whether $G$
can be edited into $C_n$ with at most $B$ edits is NP-complete.
\end{theoremA}

\noindent
The reduction is from Hamiltonian Cycle: a labeling $\pi:V(G)\to\Z_n$ carrying
$k$ edges of $G$ onto consecutive pairs has edit cost $n+m-2k$, so minimizing
the cost is maximizing the longest path realized in the cyclic order. The
completion problem is therefore hard even when the host group and its
connection set are given in advance, and all of the difficulty lies in the
labeling.

\begin{theoremB}[Degree bound, and the star]
For every connected $G$,
\[
  \gamma^{+}(G)\;\ge\;\frac{n\Delta^{*}}{2m}-1 ,
\]
computable in linear time from the degree sequence. Equality holds for the
star: $\gamma^{+}(K_{1,q})=(q-1)/2$, and $K_{1,q}$ is the unique maximizer of
both invariants among connected graphs on $n\le 7$ vertices.
\end{theoremB}

\noindent
The bound isolates the obstruction that irregularity alone imposes, and says
nothing else; in particular it is blind to whether the added edges can be
arranged into legitimate generator classes, which is precisely the difficulty
Theorem~A shows to be hard. It is attained on $89.4\%$ of our census, which
makes it usable as a triage filter.

Alongside these, Section~\ref{sec:struct} records the structural fact on which
the computations rest --- a labeling of $V(G)$ by $\Gamma$ sorts the edges
into inverse-pair classes, each class of a host contributing either a perfect
matching or a disjoint union of equal-length cycles, so the cost of a labeling
depends only on the class-count vector. Section~\ref{sec:families} determines
further families exactly, including $\gamma_{\triangle}(K_{1,q})=(2q-3)/q$,
whose limit is $2$ rather than the $3/2$ one might guess from the additive
case, and the path and grid values
$\gamma^{+}(P_n)=\gamma^{+}(P_n\square P_n)=1/(n-1)$.
Section~\ref{sec:census} reports an exhaustive certified census of all $995$
connected graphs on at most seven vertices, on which the two invariants
separate strictly $84.7\%$ of the time. Appendix~\ref{app:solver} describes a
solver for graphs beyond exhaustive range, and Appendix~\ref{app:distortion}
shows that edit count and bi-Lipschitz distortion of the completed host are
independent invariants, moving in opposite directions on stars and on paths.

\section{Related work}
\label{sec:related}

\paragraph{Isometric embedding into structured hosts.}
Which graphs embed isometrically into a hypercube goes back to Firsov
\cite{firsov1965}, with the characterization of partial cubes by a single
edge relation due to Djokovi\'c \cite{djokovic1973} and Winkler
\cite{winkler1984}; Ovchinnikov \cite{ovchinnikov2008} surveys the area. The
extension to Hamming graphs and products of complete graphs is due to Wilkeit
\cite{wilkeit1990}, with the structure theory of products developed by Graham
and Winkler \cite{graham-winkler1985} and surveyed in
\cite{imrich-klavzar2000,hammack2011}. Scale and $\ell_1$ embeddings are
treated by Shpectorov \cite{shpectorov1993} and in the monograph of Deza and
Laurent \cite{deza-laurent1997}. All of this literature enlarges the host to
achieve exactness; the present paper instead perturbs the graph.

\paragraph{The zero set and the classification programme.}
Both invariants vanish exactly when $G$ is already an abelian Cayley graph on
its own vertex set, so the zero set of $\gamma^{+}$ is the class of abelian
Cayley graphs itself. That class is the subject of an established programme:
Miklavi\v{c} and Poto\v{c}nik \cite{miklavic-potocnik2003} classified
distance-regular circulants and later posed the problem of characterizing
distance-regular Cayley graphs for arbitrary classes of groups
\cite{miklavic-potocnik2007}, with the abelian case under minimality
conditions treated by Miklavi\v{c} and \v{S}parl \cite{miklavic-sparl2014}
and the classification over $\Z_n\oplus\Z_p$ completed by Zhan, Huang and Lu
\cite{zhan-huang-lu2026}. Where that programme asks whether a graph belongs
to the zero set, the present paper assigns a numerical distance from it, and
Theorem~\ref{thm:degree} makes that distance computable in linear time from
below. The two questions meet on distance-regular graphs, where membership is
now decided for the groups covered by the classification, and
$\gamma^{+}$ measures the remainder.

\paragraph{Cayley graphs and algebraic structure.}
Background on Cayley graphs and vertex-transitive graphs is standard
\cite{godsil-royle2001}; that vertex-transitivity is strictly weaker than
being a Cayley graph is classical \cite{mckay-praeger1994}, and is one reason
the completion problem is not vacuous for regular graphs. The sum-free set
theory of Diananda and Yap \cite{diananda-yap1969}, Rhemtulla and Street
\cite{rhemtulla-street1970} and Green and Ruzsa \cite{green-ruzsa2005} is the
extremal input to the \emph{isometric} star obstruction of the companion
papers; no result of this paper uses it, and we record the connection rather
than claim it. The Smith normal form, which underlies the quotient labeling of
the companion papers, is surveyed in \cite{stanley2016}.

\paragraph{Census methodology.}
Our census enumerates the connected graphs on at most seven vertices in the
ordering of the Atlas of Graphs \cite{read-wilson1998}.

\section{Definitions}
\label{sec:def}

Throughout, $G$ is a finite connected simple graph with $n=|V(G)|\ge 2$
vertices, $m=|E(G)|$ edges and maximum degree $\Delta$. All groups are finite
and abelian, written additively. For $\Gamma$ of order $n$ and
$S=-S\subseteq\Gamma\setminus\{0\}$ generating $\Gamma$, the Cayley graph
$\Cay(\Gamma,S)$ has vertex set $\Gamma$ and edges $\{x,x+s\}$.

\begin{definition}\label{def:gamma}
Let $\pi:V(G)\to\Gamma$ be a bijection and write
$E_\pi(G)=\{\{\pi(u),\pi(v)\}:uv\in E(G)\}$. Set
\[
  \gamma^{+}(G)=\frac1m\min\bigl(|E(\Cay(\Gamma,S))|-m\bigr),
  \qquad
  \gamma_{\triangle}(G)=\frac1m\min\bigl|E_\pi(G)\,\triangle\,E(\Cay(\Gamma,S))\bigr|,
\]
the minima being over all abelian $\Gamma$ of order $n$, all bijections $\pi$
and all symmetric generating sets $S$, with the additional constraint
$E_\pi(G)\subseteq E(\Cay(\Gamma,S))$ in the first case. In both cases the
completed graph is the Cayley graph $\Cay(\Gamma,S)$ itself, not a graph
merely containing one: no edges outside the host are permitted, which is
what makes $\gamma^{+}$ finite and forces the completed graph to be
regular.
\end{definition}

Both quantities are normalized by $m$ so that they compare across graphs of
different sizes, and both vanish exactly when $G$ is itself an abelian Cayley
graph. Since a completion that only adds edges is in particular an edit,
$\gamma_{\triangle}\le\gamma^{+}$ always; Section~\ref{sec:census} shows the
inequality is strict for a large majority of graphs.

\begin{remark}
Both invariants are unchanged by translating $\pi$, since the multiset of
differences $\pi(v)-\pi(u)$ over $E(G)$ is unchanged. We may therefore always
normalize $\pi(v_0)=0$ at one chosen vertex, which reduces the search space
from $n!$ to $(n-1)!$ and is what makes the census of
Section~\ref{sec:census} exhaustive rather than heuristic.
\end{remark}

\section{The class decomposition}
\label{sec:struct}

For $c\in\Gamma\setminus\{0\}$ let $[c]=\{c,-c\}$ be its \emph{inverse-pair
class}, and let $\mathcal{C}(\Gamma)$ be the set of classes. A symmetric set
$S$ is a union of classes, so choosing $S$ is choosing a subset of
$\mathcal{C}(\Gamma)$.

\begin{lemma}[Class contribution]\label{lem:orbit}
Let $c\in\Gamma\setminus\{0\}$ and let $H_c$ be the spanning subgraph of
$\Cay(\Gamma,\{c,-c\})$. If $2c=0$ then $H_c$ is a perfect matching, with
$n/2$ edges. Otherwise $H_c$ is a disjoint union of $n/\mathrm{ord}(c)$
cycles, each of length $\mathrm{ord}(c)$, with $n$ edges in total. In both
cases we write $\mathrm{orb}(c)$ for the number of edges.
\end{lemma}

\begin{proof}
The connected components of $H_c$ are the cosets of $\langle c\rangle$, and
on each coset $H_c$ is the circulant $\Cay(\langle c\rangle,\{c,-c\})$, a
cycle of length $\mathrm{ord}(c)$ when $\mathrm{ord}(c)\ge 3$ and a single
edge when $\mathrm{ord}(c)=2$. Counting edges gives $n$ in the first case and
$n/2$ in the second.
\end{proof}

Consequently every generator class of an optimal host is either a perfect
matching, when the generator is an involution, or a $2$-factor all of whose
cycles have the same length. This is the structural reason the search space
collapses.

\begin{proposition}[Cost of a labeling]\label{prop:cost}
Fix $\Gamma$ and a bijection $\pi$, and for $c\in\mathcal{C}(\Gamma)$ let
$\chi(c)$ be the number of edges $uv\in E(G)$ with
$[\pi(v)-\pi(u)]=c$. Then for any $S\subseteq\mathcal{C}(\Gamma)$,
\[
  \bigl|E_\pi(G)\triangle E(\Cay(\Gamma,S))\bigr|
  \;=\;\sum_{c\in S}\bigl(\mathrm{orb}(c)-\chi(c)\bigr)\;+\;\sum_{c\notin S}\chi(c).
\]
The additive cost is obtained by restricting to those $S$ containing every
class with $\chi(c)>0$, in which case the second sum vanishes.
\end{proposition}

\begin{proof}
Edges of $\Cay(\Gamma,S)$ in class $c\in S$ number $\mathrm{orb}(c)$ by
Lemma~\ref{lem:orbit}, of which $\chi(c)$ already lie in $E_\pi(G)$; these
contribute $\mathrm{orb}(c)-\chi(c)$ additions. Edges of $E_\pi(G)$ in a
class outside $S$ must be deleted, contributing $\chi(c)$. Distinct classes
are disjoint, so the counts add.
\end{proof}

The cost therefore depends on $\pi$ only through the vector
$(\chi(c))_{c\in\mathcal{C}(\Gamma)}$, and on $S$ only through a subset
choice constrained by the requirement that $S$ generate $\Gamma$. Since
$|\mathcal{C}(\Gamma)|$ is small for small $n$ --- three classes for
$n\in\{6,7\}$ --- all admissible $S$ can be enumerated, and the only
remaining search is over labelings.

\section{The degree lower bound}
\label{sec:lower}

\begin{theorem}[Degree bound]\label{thm:degree}
For every connected $G$,
\begin{equation}\label{eq:degree}
  \gamma^{+}(G)\;\ge\;\frac{n\Delta^{*}}{2m}-1 ,
\end{equation}
where $\Delta^{*}$ is the least integer $d\ge\Delta$ for which $nd$ is even.
\end{theorem}

\begin{proof}
An abelian Cayley graph is regular, of degree $|S|$. A completion by
additions only cannot decrease any vertex degree, so the completed graph is
regular of some degree $d\ge\Delta$. A $d$-regular graph on $n$ vertices
exists only if $nd$ is even, so in fact $d\ge\Delta^{*}$ and the completed
graph has at least $n\Delta^{*}/2$ edges. The number of added edges is at
least $n\Delta^{*}/2-m$, and dividing by $m$ gives the claim.
\end{proof}

The bound is computable in linear time from the degree sequence alone, which
makes it usable as a triage filter: it certifies that certain graphs are far
from any abelian Cayley graph without running any search. Its practical
quality is examined in Section~\ref{sec:census}, where it is attained on
$89.4\%$ of the census and correlates with $\gamma^{+}$, the invariant it
bounds, at $r=0.851$.

The bound is sharp in a strong sense: it is exactly the obstruction that
irregularity imposes, and it says nothing else. In particular it is blind to
whether the added edges can be arranged into legitimate generator classes,
which is precisely the difficulty that makes the problem hard
(Section~\ref{sec:hard}).

\subsection{Characterizing equality}

The bound is not merely a numerical inequality; equality is an exact
structural condition.

\begin{theorem}[Equality in the degree bound]\label{thm:equality}
For a connected graph $G$,
\[
  \gamma^{+}(G)=\frac{n\Delta^{*}}{2m}-1
\]
if and only if $G$ is a spanning subgraph of some $\Delta^{*}$-regular abelian
Cayley graph on $V(G)$.
\end{theorem}

\begin{proof}
Suppose first that equality holds, and let $H=\Cay(\Gamma,S)$ be an optimal
additive completion, with $d=|S|$. Since $G\subseteq H$ and $H$ is
$d$-regular, every vertex of $G$ has degree at most $d$, so $d\ge\Delta$; and
$nd$ is even because $H$ is a $d$-regular graph on $n$ vertices, so
$d\ge\Delta^{*}$. The number of added edges is $nd/2-m\ge n\Delta^{*}/2-m$,
so $\gamma^{+}(G)\ge nd/(2m)-1\ge n\Delta^{*}/(2m)-1$. Equality in the outer
terms forces $d=\Delta^{*}$, so $H$ is a $\Delta^{*}$-regular abelian Cayley
graph containing $G$ as a spanning subgraph.

Conversely, if $G$ is a spanning subgraph of a $\Delta^{*}$-regular abelian
Cayley graph $H$ on $V(G)$, then adding the $n\Delta^{*}/2-m$ edges of
$E(H)\setminus E(G)$ is a valid completion, so
$\gamma^{+}(G)\le n\Delta^{*}/(2m)-1$; Theorem~\ref{thm:degree} supplies the
reverse inequality.
\end{proof}

\begin{remark}\label{rem:parity}
When $n\Delta$ is even we have $\Delta^{*}=\Delta$ and \eqref{eq:degree} is
the plain degree bound. When $n\Delta$ is odd --- that is, when $n$ and
$\Delta$ are both odd --- no $\Delta$-regular graph on $n$ vertices exists, so
the host is regular of degree at least $\Delta+1=\Delta^{*}$ and
Theorem~\ref{thm:equality} is vacuous for $\Delta$ itself. The gain over the
rounded form $\lceil n\Delta/2\rceil/m-1$ is genuine rather than cosmetic:
$n(\Delta+1)/2-\lceil n\Delta/2\rceil=(n-1)/2>0$ for every $n\ge3$, so
\eqref{eq:degree} is strictly stronger in exactly that case.
The parity correction is not cosmetic on the census either. Under the plain
form $n\Delta/2m-1$ the bound is attained on $476$ of the $995$ graphs
($47.8\%$); under \eqref{eq:degree} it is attained on $890$ ($89.4\%$). The
entire gain falls, as it must, on the graphs with $n$ and $\Delta$ both odd:
attainment at $n=5$ rises from $61.9\%$ to $100\%$ and at $n=7$ from $42.1\%$
to $89.7\%$, while the even orders $n=4$ and $n=6$ are unchanged at $100\%$
and $84.8\%$.
\end{remark}

Theorem~\ref{thm:equality} explains the census figure of
Section~\ref{sec:census}: the graphs attaining the bound are exactly those
that sit inside an abelian Cayley graph whose degree already matches their
own maximum degree. It also shows that the condition is not one that can be
read off the degree sequence, so the linear-time bound and the exact value
part company precisely here. Deciding membership in this class is
Problem~\ref{prob:equality}.

\section{Exactly determined families}
\label{sec:families}

\subsection{Stars}

The star is the extreme case of irregularity and, as in the isometric theory
\cite{fokam-p2}, it is the family on which the invariants are largest.

\begin{theorem}[Star obstruction]\label{thm:star}
For $q\ge 2$,
\[
  \gamma^{+}(K_{1,q})=\frac{q-1}{2},
  \qquad
  \gamma_{\triangle}(K_{1,q})\le\frac{2q-3}{q},
\]
with equality in the second for every $q\ge 2$; at $q=2$ the value is
$\gamma_{\triangle}(K_{1,2})=\gamma_{\triangle}(P_3)=1/2$. In particular
$\gamma_{\triangle}(K_{1,q})\to 2$, and not $3/2$.
\end{theorem}

\begin{proof}
Here $n=q+1$, $m=q$ and $\Delta=q$. Theorem~\ref{thm:degree} gives
$\gamma^{+}\ge (q+1)q/(2q)-1=(q-1)/2$. For the matching upper bound take
$\Gamma=\Z_{q+1}$ and $S=\Gamma\setminus\{0\}$, that is, the complete graph
$K_{q+1}=\Cay(\Z_{q+1},\Z_{q+1}\setminus\{0\})$, which contains every graph
on $q+1$ vertices. The completion adds
\[
  \binom{q+1}{2}-q=\frac{q(q+1)}{2}-q=\frac{q(q-1)}{2}
\]
edges, and dividing by $m=q$ gives exactly $(q-1)/2$. Since this meets the
lower bound, the complete graph is an optimal additive completion, not a
wasteful one. For the edit bound, delete one pendant edge and complete the
remaining $K_{1,q-1}$ together with the isolated vertex inside $K_{q+1}$;
we argue exactly, in both directions.

Let $\Cay(\Gamma,S)$ be any host of order $n=q+1$ and put $s=|S|$. Every
edge of $K_{1,q}$ meets the centre, and the centre has degree $s$ in the
host, so at most $k\le\min(q,s)$ star edges survive. The host has $ns/2$
edges, so
\[
  \bigl|E_\pi(G)\triangle E(\Cay(\Gamma,S))\bigr|
  =\underbrace{(q-k)}_{\text{deletions}}+\underbrace{\Bigl(\tfrac{ns}{2}-k\Bigr)}_{\text{additions}}
  =q+\frac{(q+1)s}{2}-2k .
\]
Since $S\subseteq\Gamma\setminus\{0\}$ we have $s\le n-1=q$, so
$k\le\min(q,s)=s$ and therefore $-2k\ge-2s$, giving
\[
  \bigl|E_\pi(G)\triangle E(\Cay(\Gamma,S))\bigr|
  \;\ge\;q+\frac{(q+1)s}{2}-2s\;=\;q+\frac{s(q-3)}{2}.
\]
Connectedness of the host forces $s\ge2$, and for $q\ge2$ the coefficient
$(q-3)/2$ is non-negative, so the right-hand side is minimised at $s=2$,
where it equals $q+(q-3)=2q-3$. Equality is attained by
$\Gamma=\Z_{q+1}$, $S=\{\pm1\}$, that is by the cycle $C_{q+1}$: place the
centre anywhere, keep the two star edges joining it to its two cycle
neighbours, delete the other $q-2$ star edges and add the remaining $q-1$
cycle edges, for $(q-2)+(q-1)=2q-3$ edits. Dividing by $m=q$ gives
$\gamma_{\triangle}(K_{1,q})=(2q-3)/q$.
\end{proof}

The optimal hosts found by the exact procedure of Section~\ref{sec:census}
are indeed cycles for every $3\le q\le 8$, which is a useful consistency
check on the implementation.

The contrast with the additive case is the substantive point: $\gamma^{+}$
grows without bound on stars, whereas $\gamma_{\triangle}$ stays below $2$.
Allowing deletions therefore does not merely improve the constant, it
changes an unbounded quantity into a bounded one. The star is the unique
global maximizer of $\gamma_{\triangle}$ in our census.

\subsection{Paths and grids}

\begin{proposition}\label{prop:path}
$\gamma^{+}(P_n)=\dfrac1{n-1}$ and
$\gamma^{+}(P_n\square P_n)=\dfrac1{n-1}$.
\end{proposition}

\begin{proof}
For the path, $m=n-1$ and adding the single edge joining the two endpoints
produces $C_n=\Cay(\Z_n,\{\pm1\})$, so $\gamma^{+}\le 1/(n-1)$; and
$\gamma^{+}>0$ because $P_n$ is not regular. For the grid, adding the $n$
wrap-around edges in each of the two directions produces the torus
$C_n\square C_n=\Cay(\Z_n\times\Z_n,\{\pm e_1,\pm e_2\})$. The grid has
$m=2n(n-1)$ edges and $2n$ edges are added, giving $1/(n-1)$ again.
Optimality follows from Theorem~\ref{thm:degree}, which for the grid gives
exactly $1/(n-1)$.
\end{proof}

That the two families share a value is a coincidence of normalization rather
than a structural identity, but it is a useful sanity check on any
implementation: both are reproduced exactly by the solver of
Section~\ref{sec:algo}, as are the values $0.2000$ for the $6\times6$ grid
and $0.1111$ for the $10\times 10$ grid.

\subsection{Graphs with vanishing invariant}

$\gamma_{\triangle}(G)=0$ if and only if $G$ is an abelian Cayley graph, and
in particular $G$ must be regular. Among the $995$ connected graphs on at
most seven vertices, exactly $14$ satisfy this, all of them regular, and
every one of them is a circulant.

\section{Universal upper bounds and the extremal graph}
\label{sec:upper}

Theorem~\ref{thm:degree} bounds $\gamma^{+}$ from below. We now bound it from
above, uniformly over all connected graphs, which settles the extremal
question for $\gamma^{+}$ and separates the two invariants asymptotically.

\begin{lemma}\label{lem:trivial}
For every connected graph $G$,
\[
  \gamma^{+}(G)\;\le\;\frac{n(n-1)}{2m}-1
  \qquad\text{and}\qquad
  \gamma_{\triangle}(G)\;\le\;1+\frac{n}{m}.
\]
\end{lemma}

\begin{proof}
The complete graph is $K_n=\Cay(\Z_n,\Z_n\setminus\{0\})$, so completing $G$
to $K_n$ by adding the $\binom{n}{2}-m$ missing edges is always a valid
additive completion; dividing by $m$ gives the first bound. For the second,
the cycle $C_n=\Cay(\Z_n,\{\pm1\})$ is an abelian Cayley graph on $n$
vertices, and editing $G$ into any fixed labelling of $C_n$ costs at most
$m+n$ edits, namely deleting every edge of $G$ and adding every edge of
$C_n$.
\end{proof}

The two bounds behave very differently, and this is the sharpest structural
difference between the invariants.

\begin{corollary}\label{cor:bounded}
$\gamma_{\triangle}(G)\le 1+n/(n-1)\le 3$ for every connected $G$, and
$\gamma_{\triangle}(G)\le 2$ whenever $G$ is not a tree. By contrast
$\gamma^{+}$ is unbounded: $\gamma^{+}(K_{1,q})=(q-1)/2\to\infty$.
\end{corollary}

\begin{proof}
Connectivity gives $m\ge n-1$, and $1+n/m$ is decreasing in $m$; a graph that
is not a tree has $m\ge n$. The values for stars are
Theorem~\ref{thm:star}.
\end{proof}

So allowing deletions does not merely lower the cost, it changes the order of
growth: the edit distance is bounded by an absolute constant while the
completion number is not. Any graph, however irregular, is within a bounded
number of edits per edge of an abelian Cayley graph; the cost of insisting on
additions alone is unbounded.

\subsection{The star is extremal for $\gamma^{+}$}

\begin{theorem}[The star maximizes $\gamma^{+}$]\label{thm:starmax}
For every connected graph $G$ on $n\ge 3$ vertices,
\[
  \gamma^{+}(G)\;\le\;\frac{n-2}{2}\;=\;\gamma^{+}(K_{1,n-1}),
\]
so the star attains the maximum of $\gamma^{+}$ over all connected graphs on
$n$ vertices.
\end{theorem}

\begin{proof}
By Lemma~\ref{lem:trivial}, $\gamma^{+}(G)\le n(n-1)/(2m)-1$. This expression
is decreasing in $m$, and connectivity gives $m\ge n-1$, so
\[
  \gamma^{+}(G)\;\le\;\frac{n(n-1)}{2(n-1)}-1\;=\;\frac{n}{2}-1
  \;=\;\frac{n-2}{2}.
\]
Theorem~\ref{thm:star} with $q=n-1$ gives
$\gamma^{+}(K_{1,n-1})=(n-2)/2$, so the value is attained.
\end{proof}

Uniqueness is a separate matter, and we prove only a reduction. Equality in
the proof above forces both $m=n-1$, so that $G$ is a tree, and equality in
Lemma~\ref{lem:trivial}, so that no abelian Cayley graph of degree less than
$n-1$ contains $G$ as a spanning subgraph. The star satisfies both: its
maximum degree is $n-1$, so $K_n$ is its only possible host. For a tree $T$
that is not a star we have $\Delta(T)\le n-2$, and it suffices to exhibit one
proper abelian Cayley host. Two are available. For $n$ even,
$\Cay(\Z_n,\Z_n\setminus\{0,n/2\})$ is $K_n$ minus a perfect matching, and
$T$ is a spanning subgraph of it under some labelling precisely when the
complement of $T$ has a perfect matching. For $n$ odd,
$\Cay(\Z_n,\Z_n\setminus\{0,\pm1\})$ is $K_n$ minus a Hamiltonian cycle, and
$T$ embeds precisely when the complement of $T$ is Hamiltonian. Uniqueness
therefore reduces to a classical question about complements of trees, which
we record as Problem~\ref{prob:unique}.

\begin{conjecture}\label{conj:starunique}
$K_{1,n-1}$ is the unique maximizer of $\gamma^{+}$ among connected graphs on
$n\ge 4$ vertices, and also the unique maximizer of $\gamma_{\triangle}$.
\end{conjecture}

The two halves are not equally supported. The first is
Theorem~\ref{thm:starmax} together with the reduction above. The second is
much stronger than anything proved here: since
$\gamma_{\triangle}(K_{1,q})=(2q-3)/q\to 2$ while
Corollary~\ref{cor:bounded} only gives $\gamma_{\triangle}\le 2$ for
non-trees, a proof would have to separate two quantities that agree in the
limit. The census of Section~\ref{sec:census} confirms both halves for
$n\le 7$ and is our only evidence for the second.

\subsection{Random graphs}

Before reading these figures, note their composition: $853$ of the $995$
graphs have exactly seven vertices, so the median, the two rates and the
correlation are, in effect, statistics of $n=7$ rather than properties of the
invariants across orders. They should be stratified by $n$ before any of them
is quoted as a general tendency.
The degree bound is attained on most of the census
(Table~\ref{tab:census-strat}), but small graphs are atypical. On random graphs the bound is worthless and the trivial upper bound
of Lemma~\ref{lem:trivial} is essentially the truth.

\begin{theorem}\label{thm:random}
Let $G\sim G(n,1/2)$. Then with high probability
$\gamma^{+}(G)=1-o(1)$; that is, completing $G$ to $K_n$ is asymptotically
optimal.
\end{theorem}

\begin{proof}
Write $N=\binom{n}{2}$. An abelian Cayley graph on the vertex set $V$ is
determined by an abelian group $\Gamma$ of order $n$, a bijection
$V\to\Gamma$, and a symmetric set $S\subseteq\Gamma\setminus\{0\}$. The
number of abelian groups of order $n$ is $n^{o(1)}$, so the number of
candidate hosts is at most $n^{o(1)}\cdot n!\cdot 2^{n}=2^{(1+o(1))n\log n}$.

Fix a host $H$ with $M$ edges. Since the $N$ potential edges of $G$ are
independent, $\Pr[G\subseteq H]=2^{-(N-M)}$. Hence the expected number of
hosts $H$ with $N-M\ge n^{3/2}$ that contain $G$ is at most
$2^{(1+o(1))n\log n-n^{3/2}}=o(1)$. By Markov's inequality, with high
probability every abelian Cayley graph containing $G$ has at least
$N-n^{3/2}$ edges.

Also with high probability $m=N/2+O(n)$. Therefore
\[
  \gamma^{+}(G)\;\ge\;\frac{N-n^{3/2}-m}{m}
  \;=\;\frac{N/2-n^{3/2}+O(n)}{N/2+O(n)}\;=\;1-O(n^{-1/2}),
\]
while Lemma~\ref{lem:trivial} gives $\gamma^{+}(G)\le N/m-1=1+O(n^{-1})$.
\end{proof}

\begin{remark}
Theorem~\ref{thm:random} should be contrasted with the degree bound, which on
$G(n,1/2)$ evaluates to
$n\Delta/(2m)-1=\Theta\bigl(\sqrt{\log n/n}\bigr)$, since
$\Delta=n/2+\Theta(\sqrt{n\log n})$ and $m=N/2+O(n)$ with high probability.
The bound therefore tends to $0$ while the truth tends to $1$: on almost all
graphs the degree bound captures none of the difficulty. Its usefulness on
the census is a feature of small, structured graphs, and the equality class
of Theorem~\ref{thm:equality} has density $o(1)$.
\end{remark}

\section{Hardness}
\label{sec:hard}

\begin{theorem}\label{thm:nphard}
Deciding, for a fixed cyclic host $C_n$ and a given budget $B$, whether
$G$ can be edited into $C_n$ with at most $B$ edits is NP-complete.
\end{theorem}

\begin{proof}
Reduce from Hamiltonian Cycle. Given $G$ on $n$ vertices with $m$ edges, a
labeling $\pi:V(G)\to\Z_n$ places the edges of $G$ into the classes of
$\Z_n$; restrict the host to $S=\{\pm1\}$, so the host is the cycle $C_n$
with $n$ edges. If the labeling maps a subgraph of $G$ with $k$ edges onto
consecutive pairs $\{i,i+1\}$, the edit cost is $(n-k)+(m-k)=n+m-2k$, since
$n-k$ host edges must be added and the other $m-k$ edges of $G$ deleted.
Minimizing the cost is therefore maximizing $k$. The edges counted by $k$
are those mapped onto consecutive pairs of the cyclic order, so they form a
subgraph of the cycle $C_n$ on $V(G)$, that is, a spanning \emph{linear
forest} of $G$: a disjoint union of paths covering every vertex, together
with the full cycle in the extremal case. A spanning linear forest with $p$
components has $n-p$ edges, so
\[
  \max_\pi k=n-p_{\min}(G),
\]
where $p_{\min}(G)$ is the least number of parts in a partition of $V(G)$
into vertex-disjoint paths --- the \emph{path partition number} $\pp(G)$ ---
except that $p_{\min}(G)=0$ when $G$ is Hamiltonian, the whole cycle then
being available in place of a spanning linear forest. The distinction matters
below: $\pp(G)\ge 1$ for every graph, so $p_{\min}$ and $\pp$ agree exactly on
the non-Hamiltonian graphs. Hence the minimum edit cost is
$m-n+2p_{\min}(G)$, and it equals $m-n$ precisely when $p_{\min}(G)=0$, that
is, exactly when $G$ has a Hamiltonian cycle. Hamiltonicity is NP-complete,
so the completion problem with fixed host is NP-hard. Membership in NP is
immediate: a labeling $\pi$ is a certificate, and its edit cost against the
fixed host is computed in time $O(n^2)$. The problem is therefore
NP-complete.
\end{proof}

The reduction uses only the cyclic host, so hardness is not an artefact of
searching over groups: it is present even when the target is completely
specified. This is what separates the completion problem from the lower bound
of Section~\ref{sec:lower}, which is linear-time.

Identifying the objective with $p_{\min}$ also transfers hardness to every
class on which Hamiltonicity is hard, so Theorem~\ref{thm:nphard} holds for
planar cubic graphs and for bipartite graphs.

\subsection{An exact cyclic value and a computable relaxation}

The identification of the objective with $p_{\min}$ cuts the other way as
well. It places the fixed-host problem inside a well-studied family of
partition parameters, and every upper bound known for that family becomes an
upper bound for $\gamma_{\triangle}$. Write
$\gamma_{\triangle}^{\mathrm{cyc}}(G)$ for the variant of
$\gamma_{\triangle}(G)$ in which the host is required to be the cycle, that
is, $\Gamma=\Z_n$ and $S=\{\pm1\}$.

\begin{corollary}[Exact value on a cyclic host]\label{cor:cyclic-exact}
Let $G$ be connected with $n\ge3$. Then
\[
  \gamma_{\triangle}^{\mathrm{cyc}}(G)=
  \begin{cases}
    \dfrac{m-n}{m} & \text{if $G$ is Hamiltonian,}\\[2.2ex]
    \dfrac{m-n+2\pp(G)}{m} & \text{otherwise,}
  \end{cases}
\]
and in particular
$\gamma_{\triangle}(G)\le\gamma_{\triangle}^{\mathrm{cyc}}(G)$.
\end{corollary}

\begin{proof}
The displayed value is the minimum edit cost computed in the proof of
Theorem~\ref{thm:nphard}, divided by $m$, with $p_{\min}$ resolved into its
two cases. The inequality holds because $\Cay(\Z_n,\{\pm1\})$ is one of the
hosts admissible in Definition~\ref{def:gamma}.
\end{proof}

Computing $\pp$ is itself NP-hard, so Corollary~\ref{cor:cyclic-exact} trades
one hard problem for another. It becomes useful when combined with an upper
bound on $\pp$ that is cheap to evaluate. Penev, Sandeep, Supraja and
Taruni~\cite{penev-sandeep-supraja-taruni2026} supply one, in terms of the
matching number $\nu(G)$: they prove
$\indpp(G)\le\ipp(G)\le n-\nu(G)$, where $\indpp$ and $\ipp$ are the induced
and isometric path partition numbers. Since every induced path is a path,
$\pp\le\indpp$, and the chain extends to $\pp$ on the left.

\begin{proposition}[Matching relaxation]\label{prop:matching-bound}
Let $G$ be connected with $n\ge3$. Then
\[
  \gamma_{\triangle}(G)\;\le\;1+\frac{n-2\nu(G)-2}{m},
\]
and the right-hand side is computable in polynomial time.
\end{proposition}

\begin{proof}
We first show $\pp(G)\le n-\nu(G)-1$. Let $M$ be a maximum matching of $G$
and $U$ the set of $M$-unsaturated vertices, so that $M\cup U$ is a path
partition of $G$ with $|M|+|U|=n-\nu(G)$ parts. Since $G$ is connected with
$n\ge3$, the matching $M$ is nonempty; fix $uv\in M$. Again by connectivity
and $n\ge3$, one of $u,v$ has a neighbour outside $\{u,v\}$, say
$vw\in E(G)$ with $w\notin\{u,v\}$. If $w\in U$, replace the parts $uv$ and
$w$ by the single path $uvw$. Otherwise $ww'\in M$ for some
$w'\notin\{u,v,w\}$, and we replace the parts $uv$ and $ww'$ by the single
path $uvww'$. Either way we obtain a path partition with $n-\nu(G)-1$ parts.

Now let $G$ be non-Hamiltonian. By Corollary~\ref{cor:cyclic-exact},
\[
  \gamma_{\triangle}(G)\le\frac{m-n+2\pp(G)}{m}
  \le\frac{m-n+2\bigl(n-\nu(G)-1\bigr)}{m}
  =1+\frac{n-2\nu(G)-2}{m}.
\]
If instead $G$ is Hamiltonian then $\nu(G)=\lfloor n/2\rfloor$, so
$n-2\nu(G)-2\ge-2$, while $\gamma_{\triangle}(G)\le(m-n)/m=1-n/m\le1-2/m$;
the bound holds in this case too. Finally, $\nu(G)$ is computable in
polynomial time by Edmonds' algorithm, and $n$ and $m$ are read off $G$.
\end{proof}

The step $\pp(G)\le n-\nu(G)-1$ is the connected case of an observation
recorded as Proposition~1.6 of the arXiv version
of~\cite{penev-sandeep-supraja-taruni2026}, where it is credited to an
anonymous referee and stated in the sharper form that the graphs attaining
$\pp(G)=n-\nu(G)$ are exactly the disjoint unions of $K_1$'s and $K_2$'s. It
does not appear in the published version, so we have given the two-line
argument above rather than cite it.

\begin{remark}[Where the relaxation is tight]\label{rem:relaxation-tight}
Proposition~\ref{prop:matching-bound} is attained on the star. For
$G=K_{1,q}$ we have $n=q+1$, $m=q$ and $\nu=1$, so the bound reads
$(2q-3)/q$; and $\pp(K_{1,q})=q-1$, since a path in a star has at most three
vertices, so Corollary~\ref{cor:cyclic-exact} gives the same value. At $q=6$
both equal $3/2$, which Table~\ref{tab:census} records as the maximum of
$\gamma_{\triangle}$ over the census, attained only by $K_{1,6}$.

The relaxation is loose in the opposite regime. On a path $P_n$ with $n$
even the bound gives $1-2/(n-1)$, whereas $\pp(P_n)=1$ and
Corollary~\ref{cor:cyclic-exact} gives $1/(n-1)$. The reason is visible in
the proof: a perfect matching wastes the whole of $\nu$ on a graph that a
single path already covers.

Note also that the unsharpened chain, which would give
$\gamma_{\triangle}\le1+(n-2\nu)/m$, is never attained by a connected graph
on at least three vertices, since by the result quoted above the only
$\pp$-extremal graphs are disjoint unions of $K_1$'s and $K_2$'s. The
subtraction of $2/m$ is therefore not a refinement of an otherwise sharp
bound but a correction of a systematic one.
\end{remark}

\begin{remark}[Why the induced and isometric variants are the wrong
invariants here]\label{rem:wrong-invariant}
Theorems~1.6 and 1.7 of~\cite{penev-sandeep-supraja-taruni2026} characterize
the graphs attaining $\ipp(G)=n-\nu(G)$, respectively
$\indpp(G)=n-\nu(G)$, as those all of whose blocks are odd complete graphs,
with one even exceptional block permitted. On such graphs the upper end of
the chain is tight while $\pp$ drops strictly below it: $K_3$ has
$\ipp=\indpp=2=n-\nu$ but $\pp=1$. Since it is $\pp$, not $\ipp$ or
$\indpp$, that measures the cost in Theorem~\ref{thm:nphard} --- the edges
counted there form a spanning linear forest, with no induced or isometric
restriction --- their extremal characterization bounds the completion cost
without ever computing it. This is what forces the separate treatment of
$\pp$ above.
\end{remark}

\subsection{Bounded edit budget}

A natural parameter is the budget itself. Let $B$ be the number of edges that
may be added, and write $N=\binom{n}{2}$.

\begin{proposition}\label{prop:xp}
Suppose that deciding whether a given graph on $n$ vertices is an abelian
Cayley graph can be done in time $T(n)$. Then deciding whether
$\gamma^{+}(G)\le B/m$ can be done in time $O(n^{2B})\cdot T(n)$.
\end{proposition}

\begin{proof}
Enumerate the at most $\sum_{i\le B}\binom{N-m}{i}=O(n^{2B})$ sets of at most
$B$ non-edges, add each in turn, and test the resulting graph.
\end{proof}

The recognition step is unconditional when the host is required to be cyclic.
Evdokimov and Ponomarenko \cite{evdokimov-ponomarenko2004} give a
polynomial-time algorithm deciding whether a given graph is a circulant, by way
of the theory of Schur rings over a cyclic group and coherent configurations.
Combining this with Proposition~\ref{prop:xp} gives the following.

\begin{corollary}\label{cor:xp-cyclic}
For a fixed budget $B$, deciding whether $G$ can be completed to a circulant by
adding at most $B$ edges is solvable in polynomial time.
\end{corollary}

This places the cyclic-host problem in $\mathrm{XP}$ with respect to $B$
unconditionally, and the general problem $\mathrm{XP}$ conditionally
on the recognition step, but neither in $\mathrm{FPT}$: the exponent grows with
$B$. Whether it is fixed-parameter tractable is open, and the obstruction is
visible at the bottom of the parameter range. Any $\mathrm{FPT}$ algorithm
would in particular decide the case $B=0$ in polynomial time, and that is the
question of recognising abelian Cayley graphs, equivalently of deciding
whether $\operatorname{Aut}(G)$ contains a regular abelian subgroup. We are
not aware of a polynomial algorithm for that problem in general, so
Problem~\ref{prob:fpt} should be read as asking about the recognition problem
first.

\section{The census}
\label{sec:census}

We computed $\gamma^{+}$ and $\gamma_{\triangle}$ exactly for all $995$
connected graphs on $2\le n\le 7$ vertices, one representative per
isomorphism class, enumerated in the ordering of the Atlas of Graphs
\cite{read-wilson1998}. All statistics below are over these $995$
isomorphism classes, each counted once. Exactness rests on the three
reductions of Section~\ref{sec:struct}: translation normalization leaves
$(n-1)!$ labelings, Proposition~\ref{prop:cost} reduces the cost of a
labeling to its class-count vector, and the admissible generating sets are
enumerated exhaustively. No heuristic enters, so every value below is a
certificate.

\begin{convention}\label{conv:connected}
Throughout, a completion target $\Cay(\Gamma,S)$ is required to be connected,
that is, $S$ generates $\Gamma$. For $\gamma^{+}$ this is automatic, since the
host contains the connected graph $G$ as a spanning subgraph. For
$\gamma_{\triangle}$ it is a genuine restriction, because deletions can
disconnect, and it is the convention under which every figure in this section
was computed. Dropping it changes the census: the number of graphs with
$\gamma_{\triangle}<\gamma^{+}$ rises from $843$ to $845$, the two additional
graphs being ones whose cheapest edit target is disconnected. No other
statistic reported here is affected.
\end{convention}

Table~\ref{tab:space} records the resulting search space: for
every order $n\le 7$ except $n=4$ there is a unique abelian group, and the
number of inverse-pair classes never exceeds three, so the enumeration over
generating sets is negligible and the cost is dominated by the $(n-1)!$
labelings.

\begin{table}[t]
\centering
\caption{The exhaustive search space per order. $|\mathcal{C}(\Gamma)|$ is
listed for each abelian group of that order.}
\label{tab:space}
\begin{tabular}{rrcr}
\toprule
$n$ & abelian groups & $|\mathcal{C}(\Gamma)|$ & labelings $(n-1)!$\\
\midrule
$2$ & $1$ & $1$ & $1$\\
$3$ & $1$ & $1$ & $2$\\
$4$ & $2$ & $2,\,3$ & $6$\\
$5$ & $1$ & $2$ & $24$\\
$6$ & $1$ & $3$ & $120$\\
$7$ & $1$ & $3$ & $720$\\
\bottomrule
\end{tabular}
\end{table}

\begin{table}[t]
\centering
\caption{The census of all $995$ connected graphs on at most seven vertices.
All values exact.}
\label{tab:census}
\begin{tabular}{lr}
\toprule
quantity & value\\
\midrule
graphs in census & $995$\\
median $\gamma_{\triangle}$ & $0.3846$\\
mean $\gamma_{\triangle}$ & $0.4006$\\
graphs with $\gamma_{\triangle}<\gamma^{+}$ & $843$ \ ($84.7\%$)\\
graphs with $\gamma_{\triangle}=0$ & $14$\\
maximum $\gamma_{\triangle}$ & $1.5000$, attained only by $K_{1,6}$\\
degree bound \eqref{eq:degree} attained & $890$ \ ($89.4\%$)\\
\quad under the plain form $n\Delta/2m-1$ & $476$ \ ($47.8\%$)\\
degree bound violated & $0$\\
Pearson $r$(bound, $\gamma^{+}$) & $0.8506$\\
Pearson $r$(bound, $\gamma_{\triangle}$) & $0.7411$\\
\bottomrule
\end{tabular}
\end{table}

\begin{table}[t]
\centering
\caption{The same census stratified by order. The pooled figures of
Table~\ref{tab:census} are dominated by the $853$ graphs on seven vertices,
and every rate varies sharply with $n$. Attainment is of \eqref{eq:degree};
the plain form $n\Delta/2m-1$ is given in brackets where it differs, and
differs only at odd $n$.}
\label{tab:census-strat}
\begin{tabular}{rrrrrrr}
\toprule
$n$ & graphs & med.\ $\gamma^{+}$ & med.\ $\gamma_{\triangle}$ &
$\gamma_{\triangle}<\gamma^{+}$ & bound attained & $r$(bound, $\gamma^{+}$)\\
\midrule
$4$ & $6$   & $0.267$ & $0.267$ & $0.0\%$  & $100.0\%$            & $1.000$\\
$5$ & $21$  & $0.667$ & $0.400$ & $61.9\%$ & $100.0\%$ \ ($61.9$) & $1.000$\\
$6$ & $112$ & $0.500$ & $0.333$ & $73.2\%$ & $84.8\%$             & $0.925$\\
$7$ & $853$ & $0.750$ & $0.385$ & $87.7\%$ & $89.7\%$ \ ($42.1$)  & $0.827$\\
\midrule
all & $995$ & $0.615$ & $0.385$ & $84.7\%$ & $89.4\%$ \ ($47.8$)  & $0.851$\\
\bottomrule
\end{tabular}
\end{table}

Three features of Table~\ref{tab:census} deserve comment.

\emph{Deletions usually help.} For $84.7\%$ of the census the edit distance
is strictly smaller than the completion number. Restricting to additions is
thus not a harmless simplification, and the two invariants should be reported
separately. Figure~\ref{fig:census}B shows the gap widening with $n$.

\emph{The star is extremal.} $K_{1,6}$ is the unique maximizer of
$\gamma_{\triangle}$, at exactly $3/2$, matching Theorem~\ref{thm:star} at
$q=6$. The star is simultaneously the family on which the isometric theory of
\cite{fokam-p2} is hardest and the family on which completion is most
expensive, which is not a priori obvious: the two frameworks fail on the same
object for the same underlying reason, its irregularity.

\emph{The linear-time bound is informative but not sufficient.} Bound
\eqref{eq:degree} is attained on $89.4\%$ of the census and correlates with
$\gamma^{+}$ at $r=0.851$, so it is a genuinely useful triage filter. It is
nonetheless not a substitute for the search: it is missed on $105$ graphs,
and Theorem~\ref{thm:random} shows that on almost all graphs it captures none
of the difficulty at all, so its performance here is a fact about small
structured graphs and not a general one.

\emph{The pooled rates are dominated by one order.} Since $853$ of the $995$
graphs have $n=7$, every pooled figure in Table~\ref{tab:census} is close to
its $n=7$ value. Table~\ref{tab:census-strat} stratifies, and the variation is
large: separation runs from $0\%$ at $n=4$ to $87.7\%$ at $n=7$, and
$r(\text{bound},\gamma^{+})$ falls monotonically from $1.000$ to $0.827$ as
$n$ grows. Both trends point the same way, towards the asymptotic regime of
Theorem~\ref{thm:random}, and neither should be read as a constant of the
problem.

\begin{figure}[t]
\centering
\includegraphics[width=\linewidth]{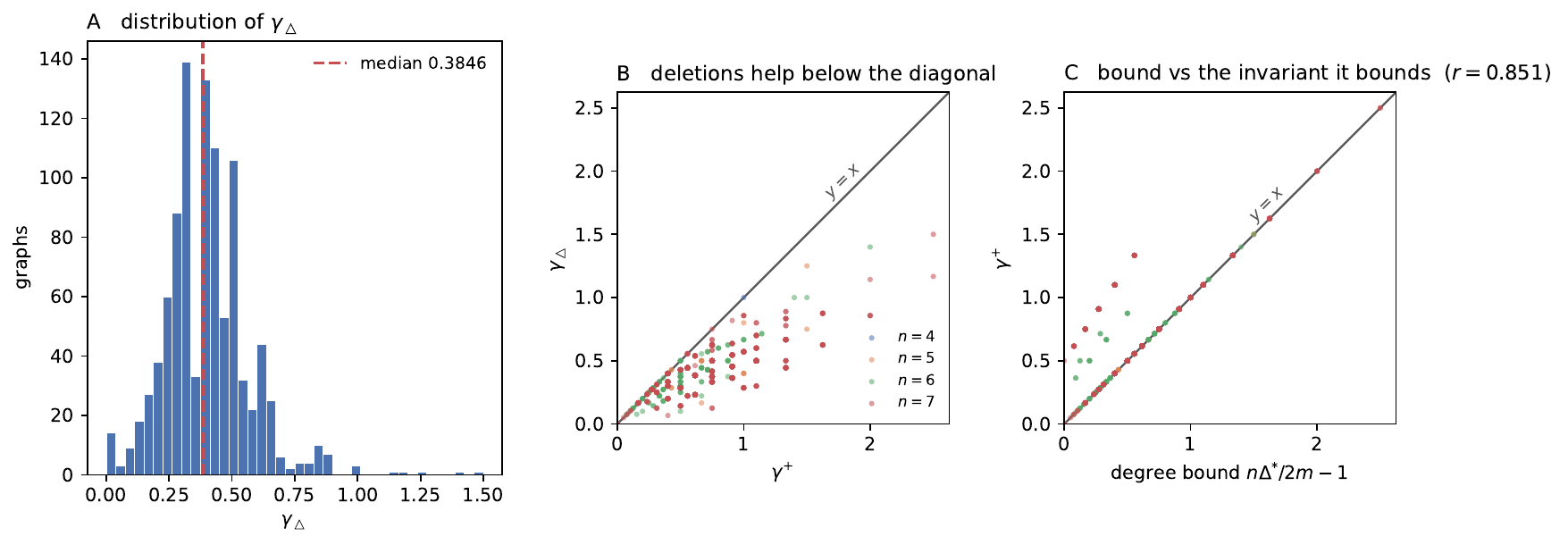}
\caption{The exact census of all $995$ connected graphs on at most seven
vertices. \textbf{A}~Distribution of $\gamma_{\triangle}$, with median
$0.3846$. \textbf{B}~$\gamma_{\triangle}$ against $\gamma^{+}$, coloured by
$n$; the $84.7\%$ of graphs for which deletions strictly help are those
strictly below the drawn diagonal $y=x$. \textbf{C}~The linear-time bound
\eqref{eq:degree} against $\gamma^{+}$, the invariant it actually bounds, with
Pearson $r=0.851$. Every point lies on or above $y=x$, as
Theorem~\ref{thm:degree} requires; the $890$ points on the diagonal are the
equality class of Theorem~\ref{thm:equality}. Panels B and C share equal axes
and aspect.}
\label{fig:census}
\end{figure}

\begin{remark}[The degree bound does not bound $\gamma_{\triangle}$]\label{rem:panelC}
Theorem~\ref{thm:degree} is a lower bound on $\gamma^{+}$, and since
$\gamma_{\triangle}\le\gamma^{+}$ it transports no lower bound to
$\gamma_{\triangle}$. The distinction is not academic, and the witness is
already in this paper. Let $G=\overline{C_3\cup C_4}$, the complement of the
disjoint union of a triangle and a quadrilateral. Then $G$ is $4$-regular on
$n=7$ vertices with $m=14$ edges, so $\Delta^{*}=\Delta=4$ and the bound of
Theorem~\ref{thm:degree} evaluates to $7\cdot4/(2\cdot14)-1=0$. But
$\gamma_{\triangle}(G)>0$: the only abelian group of order $7$ is $\Z_7$, and
the only $4$-regular circulant on $\Z_7$ up to isomorphism is
$\overline{C_7}$, which is connected, whereas $C_3\cup C_4$ is not; so $G$ is
not an abelian Cayley graph and some edit is required. Any plot of the degree
bound against $\gamma_{\triangle}$ therefore has points above the diagonal,
and a caption asserting otherwise would be false. This is why
Figure~\ref{fig:census}C plots the bound against $\gamma^{+}$ and not against
$\gamma_{\triangle}$; for the record the two correlations differ appreciably,
$r=0.851$ against $\gamma^{+}$ and $r=0.741$ against $\gamma_{\triangle}$.
\end{remark}

\section{Conclusion and open problems}

We have introduced two invariants measuring how far a graph is from being an
abelian Cayley graph on its own vertex set, shown the fixed-host problem
NP-hard by reduction from Hamiltonian Cycle, proved a linear-time degree lower
bound and characterized its equality case
(Theorem~\ref{thm:equality}), determined stars, paths and grids exactly, and
computed certified exact values for all $995$ connected graphs on at most
seven vertices. Two further results delimit the invariants from above: the
star maximizes $\gamma^{+}$ (Theorem~\ref{thm:starmax}) while
$\gamma_{\triangle}$ is bounded by an absolute constant
(Corollary~\ref{cor:bounded}), and on $G(n,1/2)$ the trivial completion to
$K_n$ is asymptotically optimal (Theorem~\ref{thm:random}), so the degree
bound, attained on $89.4\%$ of the census, captures none of the difficulty on
almost all graphs. Several questions remain.

\begin{problem}[Uniqueness of the extremal graph]\label{prob:unique}
Theorem~\ref{thm:starmax} shows the star attains the maximum of
$\gamma^{+}$. Prove Conjecture~\ref{conj:starunique}. By the reduction
following Theorem~\ref{thm:starmax} it suffices, for $n$ even, to show that
the complement of every non-star tree on $n$ vertices has a perfect matching,
and for $n$ odd that it is Hamiltonian. The corresponding statement for
$\gamma_{\triangle}$ is harder and needs a different argument, since
Corollary~\ref{cor:bounded} and the star value agree in the limit.
\end{problem}

\begin{problem}[Recognising the equality class]\label{prob:equality}
Theorem~\ref{thm:equality} characterizes equality in the degree bound as
containment in a $\Delta$-regular abelian Cayley host, which is not a
condition one can read off the degree sequence. Find a combinatorial
characterization, or show that deciding it is NP-hard.
\end{problem}

\begin{problem}[Parameterized complexity]\label{prob:fpt}
Is the Cayley completion problem fixed-parameter tractable in the edit budget
$B$? Proposition~\ref{prop:xp} gives only $\mathrm{XP}$, and the case $B=0$
is the recognition problem for abelian Cayley graphs. Settle the complexity
of that recognition problem first.
\end{problem}

\begin{problem}[Sparse random graphs]\label{prob:sparse}
Theorem~\ref{thm:random} settles $G(n,1/2)$, where the answer is that the
trivial completion to $K_n$ is asymptotically optimal. Determine
$\gamma^{+}(G(n,p))$ for $p=p(n)\to 0$, where the counting argument used
there no longer dominates and the degree bound may become the truth.
\end{problem}

\begin{problem}[Cubic graphs in the census]\label{prob:cubic}
Four cubic graphs in Table~\ref{tab:zoo} have edit count exactly $n/2$, while
the equally cubic Petersen and Desargues graphs do not. What distinguishes
them?
\end{problem}

\begin{problem}[Agreement of the two invariants]\label{prob:separate}
The additive and edit invariants separate on $84.7\%$ of the census.
Characterize the graphs on which $\gamma^{+}$ and $\gamma_{\triangle}$ agree.
\end{problem}

\begin{problem}[Distortion frontier]\label{prob:distortion}
Appendix~\ref{app:distortion} shows that edit count and metric distortion are
independent measurements and exhibits the two endpoints of a host-size versus
fidelity frontier. How does distortion decay with host order in general?
\end{problem}

\section*{Data availability and reproducibility}
All values reported here are computed by the exact procedure of
Section~\ref{sec:census} for $n\le 7$ and by the solver of
Appendix~\ref{app:solver} beyond it. Every value for $n\le 7$ is certified in
both directions: an explicit labeling and connection set witnesses the upper
bound, and exhaustion over all abelian groups of order $n$ and all symmetric
connection sets witnesses the lower bound. Code, certificates and the full
census are archived at \texttt{doi:10.5281/zenodo.21852006}; graph labels
follow the Atlas of Graphs \cite{read-wilson1998}.

\section*{Declaration on the use of artificial intelligence}
The authors declare that the artificial-intelligence assistant Claude
(Anthropic) was used during the preparation of this manuscript, in two roles:
support with the implementation, debugging and reproducibility of the search,
census and certification software; and language and editing support in
drafting. All definitions, theorems, proofs and their verification, together
with the conception, scientific direction and conclusions of this work, are
the authors' own. The authors have reviewed the entire manuscript and take
full responsibility for its content.

\appendix
\section{Solving beyond the exhaustive range}\label{app:solver}
\label{sec:algo}

For $n$ beyond the census the factorial search is infeasible and we use local
search over labelings, retaining the exact treatment of the generating set:
given a labeling, the optimal $S$ is still computed by enumeration via
Proposition~\ref{prop:cost}, so the only heuristic component is the choice of
labeling.

\paragraph{Calibration against certified optima.}
Because the census of Section~\ref{sec:census} supplies exact values, the
solver can be calibrated rather than merely trusted. On a uniformly random
sample of $200$ of the $995$ census graphs, run with four restarts and no
instance-specific tuning, the solver returned the certified optimum on
$200$ of $200$ instances. This does not constitute an approximation
guarantee, and none is claimed; it does mean that on the entire range where
ground truth is available the heuristic is not observed to lose anything,
which is the relevant evidence for the upper bounds reported beyond that
range. The values for paths, cycles, tori and grids quoted in
Section~\ref{sec:families} are likewise reproduced exactly.

\paragraph{Moves and seeds.}
Two moves are used: transposition of the labels of two vertices, and reversal
of a contiguous segment of the label order. The second matters because a
circular seed that is correct except for a reversed arc is repaired by one
segment reversal but requires many transpositions. Seeds are the spectral
circle (and its rank-$d$ torus generalization), a Cartesian-product labeling
for lattice-like graphs, and random restarts.

\paragraph{Acceptance.}
The single design choice that affects quality is whether equal-cost moves are
accepted. Proposition~\ref{prop:cost} implies that cost is a function of the
class-count vector alone, so the landscape is organized into large isocost
plateaus; a search that accepts only strict improvements halts on the first
plateau it meets. Table~\ref{tab:ablation} compares plateau descent with
strict descent under identical seeds, budget and move set.

\begin{table}[t]
\centering
\caption{Acceptance-rule ablation: edits found, minimum and mean over eight
independent seeds at a fixed budget of three restarts and $600$ iterations.
Both rules reach the same optimum given enough restarts; plateau descent
reaches it more reliably per run.}
\label{tab:ablation}
\begin{tabular}{lcccc}
\toprule
& \multicolumn{2}{c}{plateau descent} & \multicolumn{2}{c}{strict descent}\\
\cmidrule(lr){2-3}\cmidrule(lr){4-5}
graph & min & mean & min & mean\\
\midrule
Chv\'atal & $4$ & $4.50$ & $4$ & $6.25$\\
Frucht & $6$ & $7.00$ & $6$ & $7.25$\\
Heawood & $7$ & $7.25$ & $7$ & $9.75$\\
\bottomrule
\end{tabular}
\end{table}

The effect is real but should not be overstated: with a sufficient number of
restarts both rules locate the same best value on all three instances, and
the advantage of plateau descent is a lower mean cost per run rather than a
better optimum. We report it as a variance reduction, not as a change in what
is reachable. Taken together with the calibration above, the practical
reading is that the search is reliable on the range where it can be checked,
and that the acceptance rule governs how many restarts are needed rather
than which values are attainable.

\paragraph{Standard graphs.}
Table~\ref{tab:zoo} collects values for standard graphs. The value $1/3$
recurs across five of them, and the pattern is sharper than it first
appears: for the four cubic instances the edit count is exactly $n/2$, and
since a cubic graph has $m=3n/2$ this forces $\gamma_{\triangle}=1/3$
automatically. The recurrence is nevertheless partial rather than
structural, because the Petersen and Desargues graphs are cubic as well and
do not follow it, at $7$ edits against $n/2=5$ and $12$ against $n/2=10$;
and the icosahedral graph attains $1/3$ while being $5$-regular, with $10$
edits rather than $n/2=6$. We record the pattern as an observation with
known exceptions, not as a conjecture.

\begin{table}[t]
\centering
\caption{Upper bounds on $\gamma_{\triangle}$ for standard graphs, with the
optimal host group found. Hosts are given in primary decomposition, so that
for instance $\Z_2\times\Z_5$ denotes $\Z_{10}$.}
\label{tab:zoo}
\begin{tabular}{lrrrl}
\toprule
graph & $n$ & $m$ & edits & $\gamma_{\triangle}\le$\\
\midrule
Chv\'atal & $12$ & $24$ & $4$ & $0.1667$\\
Frucht & $12$ & $18$ & $6$ & $0.3333$\\
Icosahedral & $12$ & $30$ & $10$ & $0.3333$\\
Heawood & $14$ & $21$ & $7$ & $0.3333$\\
M\"obius--Kantor & $16$ & $24$ & $8$ & $0.3333$\\
Pappus & $18$ & $27$ & $9$ & $0.3333$\\
Petersen & $10$ & $15$ & $7$ & $0.4667$\\
Desargues & $20$ & $30$ & $12$ & $0.4000$\\
\bottomrule
\end{tabular}
\end{table}

\section{Distortion is a second, independent invariant}\label{app:distortion}
\label{sec:distortion}

The invariants of Definition~\ref{def:gamma} count edited edges. A different
and equally natural question is how faithfully the host reproduces the metric
of $G$, and the two answers are not interchangeable. This section makes the
distinction precise, because the numbers below show that they can move in
opposite directions.

\begin{definition}\label{def:distortion}
Let $f:V(G)\to\Gamma$ be injective, $|\Gamma|=N\ge n$, and let $d_H$ be the
graph metric of $H=\Cay(\Gamma,S)$. The \emph{bi-Lipschitz distortion} of $f$
is
\[
  c(f)\;=\;\max_{u\ne v}\frac{d_H(f(u),f(v))}{d_G(u,v)}
         \cdot\max_{u\ne v}\frac{d_G(u,v)}{d_H(f(u),f(v))},
\]
and $c_N(G)$ denotes its minimum over all abelian $\Gamma$ of order $N$, all
symmetric generating sets and all injective $f$.
\end{definition}

Then $c_N(G)\ge 1$ always, with $c_N(G)=1$ exactly when $G$ embeds
isometrically into a host of order $N$; the least such $N$ is the invariant
$\nu(G)$ of \cite{fokam-p2}. Note that $c_N$ is defined for every $N\ge n$,
whereas $\gamma^{+}$ and $\gamma_{\triangle}$ are defined only at $N=n$,
where the vertex set is fixed and editing is meaningful. The two families of
invariants therefore agree on no common normalization and must not be
compared numerically; $\gamma$ is a fraction of edges, $c_N$ a ratio of
distances.

\begin{proposition}[The two invariants are not monotonically related]
\label{prop:anticorr}
On stars, $\gamma^{+}(K_{1,q})=(q-1)/2$ is unbounded in $q$ while
$c_{n}(K_{1,q})=2$ for every $q$. On paths, $\gamma^{+}(P_n)=1/(n-1)$ tends
to $0$ while $c_{n}(P_n)$ grows linearly in $n$. Hence neither invariant
bounds the other.
\end{proposition}

\begin{proof}[Computation]
The $\gamma^{+}$ values are Theorem~\ref{thm:star} and
Proposition~\ref{prop:path}. For the star, the leaf images of any injective
$f$ into a host of order $n=q+1$ cannot all be pairwise nonadjacent while
each is adjacent to the centre unless the host is complete, so the leaf-leaf
distances collapse to $1$ against $d_G=2$, or expand; direct minimisation
gives $c_n=2$, confirmed exhaustively for $q\le 6$ and by search for
$q\le 10$. For the path, exhaustive minimisation over all abelian groups,
generating sets and placements gives $c_6(P_6)=4$ and $c_8(P_8)=5$, and
search gives $\lceil n/2\rceil+1$ for $n\le 12$.
\end{proof}

\begin{table}[t]
\centering
\caption{The two invariants at the completion endpoint $N=n$ move in opposite
directions. Values of $c_n$ are exhaustively certified for $n\le 8$ and are
search upper bounds beyond.}
\label{tab:anticorr}
\begin{tabular}{lrr@{\qquad}lrr}
\toprule
\multicolumn{3}{c}{stars $K_{1,q}$} & \multicolumn{3}{c}{paths $P_n$}\\
\cmidrule(lr){1-3}\cmidrule(lr){4-6}
 & $\gamma^{+}$ & $c_n$ &  & $\gamma^{+}$ & $c_n$\\
\midrule
$q=3$ & $1.0$ & $2$ & $n=4$ & $0.333$ & $3$\\
$q=5$ & $2.0$ & $2$ & $n=6$ & $0.200$ & $4$\\
$q=7$ & $3.0$ & $2$ & $n=8$ & $0.143$ & $5$\\
$q=9$ & $4.0$ & $2$ & $n=10$ & $0.111$ & $6$\\
\bottomrule
\end{tabular}
\end{table}

The reading of Table~\ref{tab:anticorr} is that a graph can be expensive to
complete yet metrically well behaved once completed, and cheap to complete
yet metrically ruined. The star is the first case: many edges must be added,
but the resulting host never stretches a distance by more than a factor of
two. The path is the second: a single added edge suffices, yet identifying
the two endpoints of a path is exactly what destroys its metric, since the
two extremities become adjacent. This is the familiar wraparound artefact of
circular convolution, appearing here as an invariant.

\paragraph{The frontier between the two frameworks.}
The completion problem and the isometric problem of
\cite{fokam-p1,fokam-p2} are the two extreme points of one trade-off. Score a
scheme by the pair (host order $N$, distortion). Completion sits at $N=n$
with whatever distortion the completed host happens to have; isometric
embedding sits at $N=\nu(G)$ with distortion $1$. Classical signal
processing has occupied the interior for decades without naming it:
zero-padding embeds $P_n$ into $C_{2n-2}$ with distortion $1$, and symmetric
extension underlies the discrete cosine transform. For the path the whole
frontier is available in closed form: consecutive placement into the cycle
$\Cay(\Z_N,\{\pm1\})$ gives
\[
  c_N(P_n)=\frac{n-1}{N-n+1}\qquad (n<N\le 2n-2),
\]
verified exhaustively to be optimal for every $N>n$ in the range checked, and
reaching $1$ exactly at the zero-padding order $N=2n-2$. For the star the
frontier is instead a step: $c_N(K_{1,q})=2$ for all $N<2q$ and $c_N=1$ at
$N=2q=\nu(K_{1,q})$, the exact isometric order of \cite{fokam-p2}. So the
decay of distortion with host size is smooth for one family and a sharp
threshold for the other, and we do not know which behaviour is typical. A
systematic study of this frontier is left to future work; we record here only
that the endpoints are the objects of this paper and of \cite{fokam-p2}, and
that they are genuinely different measurements.

\end{document}